\documentclass[submission,copyright,creativecommons]{eptcs}
\providecommand{\event}{WPTE 2017} 

\usepackage{etex}      

\usepackage{breakurl}             
\usepackage{underscore}           

\usepackage{amsmath}		
\usepackage{amssymb}		
\usepackage{cmll}
\usepackage{stmaryrd}
\usepackage{proof}
\usepackage{listings}
\usepackage[all]{xy}

\usepackage{graphicx}		
\usepackage{wrapfig}
\usepackage{color}		

\usepackage{amsthm}		
\theoremstyle{definition}
\newtheorem{definition}{Definition}[section]

\newtheorem*{notation*}{Notation}
\theoremstyle{plain}
\newtheorem{proposition}[definition]{Proposition}
\newtheorem{lemma}[definition]{Lemma}
\newtheorem{theorem}[definition]{Theorem}

\newcommand{\abs}[3]{\lambda {#1}^{#2}. #3}

\newcommand{\FV}{\mathit{FV}}
\newcommand{\plug}[2]{#1 \langle #2 \rangle}
\newcommand{\emptyCtxt}{\langle \cdot \rangle}

\newcommand{\appGen}[2]{#1 \mathbin{\$} #2}
\newcommand{\appLazy}[2]{#1 \mathbin{@} #2}
\newcommand{\appLR}[2]{#1 \mathbin{\overrightarrow{@}} #2}
\newcommand{\appRL}[2]{#1 \mathbin{\overleftarrow{@}} #2}
\newcommand{\W}[1]{\llparenthesis #1 \rrparenthesis}

\renewcommand{\oc}{\mathop{!}}

\newcommand{\twoheadmapsto}{
\mathrel{\ooalign{$\twoheadrightarrow$\cr%
\kern-.15ex\raisebox{.2ex}{\scalebox{1}[0.8]{$\shortmid$}}\cr}}}
\newcommand{\longtwoheadmapsto}{
\mathrel{\ooalign{$\longtwoheadrightarrow$\cr%
\kern-.15ex\raisebox{.2ex}{\scalebox{1}[0.8]{$\shortmid$}}\cr}}}

\newcommand{\up}{\mathord{\uparrow}}
\newcommand{\dn}{\mathord{\downarrow}}

\newcommand{\lb}[1]{\mathit{#1}}

\newcommand{\Init}{\mathit{Init}}
\newcommand{\Final}{\mathit{Final}}

\newcommand{\ev}{\mathit{Eval}}
\newcommand{\ex}{\mathit{Exec}}

\title{
Efficient Implementation of Evaluation Strategies \\
via Token-Guided Graph Rewriting
}
\author{Koko Muroya \quad\qquad Dan R. Ghica
\institute{University of Birmingham, UK}
\email{\{k.muroya,d.r.ghica\}@cs.bham.ac.uk}
}
\def\titlerunning{
Efficient Implementation of
Evaluation Strategies via Token-Guided Graph Rewriting
}
\def\authorrunning{K. Muroya and D.\,R. Ghica}
\begin{document}
\maketitle

\begin{abstract}
 In implementing evaluation strategies of the lambda-calculus, both correctness and 
 efficiency of implementation are valid concerns.
 While the notion of correctness is determined by the evaluation
 strategy, regarding efficiency there is a larger design space that can be explored, in
 particular the trade-off between space  versus time efficiency.
 We contributed to the study of this trade-off by the 
 introduction of an abstract machine for  call-by-need, inspired by 
 Girard's Geometry of Interaction, a machine combining token passing and graph rewriting. 
 This work presents an extension of the machine, to
 additionally accommodate  left-to-right and right-to-left
 call-by-value strategies.
 We show soundness and completeness of the extended machine with
 respect to each of the call-by-need and two call-by-value strategies.
 Analysing time cost of its execution classifies the machine as
 ``efficient'' in Accattoli's taxonomy of abstract machines.
\end{abstract}


\section{Introduction}

The lambda-calculus is a simple yet rich model of computation, relying on 
a single mechanism to activate a function in
computation, beta-reduction, that replaces function arguments with
actual input.
While in the lambda-calculus itself 
beta-reduction can be applied in an unrestricted way, it is evaluation strategies that determine
the way beta-reduction is applied when the lambda-calculus is used as a programming language.
Evaluation strategies often imply how intermediate results are copied, discarded,
cached or reused.
For example, everything is repeatedly evaluated as many times as
requested in the call-by-name strategy.
In the call-by-need strategy, once a function requests its input, the
input is evaluated and the result is cached for later use.
The call-by-value strategy evaluates function input and caches the
result even if the function does not require the input.

The implementation of any evaluation strategy must be
correct, first of all, i.e.\ it has to produce results as stipulated by the
strategy.
Once correctness is assured, the next concern is efficiency.
One may prefer better space efficiency, or better time
efficiency, and it is well known that one can be traded off for the other.
For example, time efficiency can be improved by caching more
intermediate results, which increases space cost.
Conversely, bounding space requires repeating computations, which adds
to the time cost.
Whereas  correctness is well defined for any evaluation strategy, there
is a certain freedom in managing efficiency.
The challenge here is how to produce a unified framework which is flexible 
enough to analyse and guide the choices required by this trade-off.
Recent studies by Accattoli et al.\
\cite{AccattoliDL16,AccattoliBM14,Accattoli16} clearly establish classes of efficiency
for a given evaluation strategy.
They characterise efficiency by means of the
number of beta-reduction applications required by the strategy, and
introduce two efficiency classes, namely ``efficient''
and ``reasonable.''
The expected efficiency of an abstract machine gives us a starting point to quantitatively
analyse the trade-offs required in an implementation.


We employ Girard's Geometry of Interaction (GoI)
\cite{Girard89GoI1}, a semantics of linear logic proofs, 
as a framework for  studying the trade-off
between time and space efficiency.
In particular we focus on GoI-style abstract machines for the
lambda-calculus, pioneered by Danos and Regnier~\cite{DanosR96} and
Mackie~\cite{Mackie95}.
These machines evaluate a term of the lambda-calculus by 
translating the term to a graph, a network of simple transducers, which 
executes by passing a data-carrying token around.

The token simulates graph rewriting without actually rewriting, which
is in fact a particular instance of the trade-off we mentioned above.
The token-passing machines keep the underlying graph fixed
and use the data stored in the token to route it.
They therefore favour space efficiency at the cost of time
efficiency.
The same computation is repeated when, instead, intermediate
results could have been cached by saving copies of certain sub-graphs representing values.

Our intention is to lift the GoI-style token passing to a framework to
analyse the trade-off of efficiency, by strategically interleaving it with
graph rewriting.
The key idea is that the token holds control over graph rewriting, by visiting
redexes and triggering rewrite rules.
Graph rewriting offers fine control over caching and sharing
intermediate results, however fetching cached results can increase the size of the
graph.
In short, introduction of graph rewriting sacrifices  space
 while favouring time
efficiency.
We expect the flexibility given by a fine-grained control over interleaving will enable 
a careful balance between
space  and time efficiency.

This idea was first introduced in our previous work~\cite{MuroyaG17},
by
developing an abstract machine that interleaves token passing with as
much graph rewriting as possible.
We showed the resulting graph-rewriting abstract machine implements
call-by-need evaluation, and it is classified as ``efficient'' in terms
of time.
We further develop this idea by proposing an extension of the graph-rewriting abstract
machine, to accommodate other evaluation strategies, namely 
left-to-right and right-to-left call-by-value.
In our framework, both call-by-value strategies involve similar tactics for caching
intermediate results as the call-by-need strategy, with the only
difference being the timing of cache creation.

\emph{Contributions.}
We extend the token-guided graph-rewriting abstract
machine for the call-by-need strategy~\cite{MuroyaG17} to the
left-to-right and right-to-left call-by-value strategies.
The presentation of the machine is revised by using term
graphs instead of proof nets~\cite{Girard87LL}, to make clearer sense
of evaluation strategies in the graphical representation of terms.
The extension is by introducing nodes that correspond to
different evaluation strategies, rather than modifying the behaviour
of existing nodes to suite different evaluation strategy demands.
We prove the soundness and completeness of the extended machine with
respect to the three evaluation strategies separately, using a ``sub-machine'' semantics, where the word
`sub' indicates both a focus on substitution and its status as an 
intermediate representation.
The sub-machine semantics is based on Sinot's ``token-passing''
semantics~\cite{Sinot05,Sinot06} that makes explicit the two main tasks of
abstract machines: searching  redexes and substituting  variables.
The time-cost analysis classifies the machine as ``efficient'' in
Accattoli's taxonomy of abstract machines~\cite{Accattoli16}.
Finally, an on-line visualiser is implemented, in which the machine
can be executed on arbitrary closed lambda-terms\footnote{
Link to the on-line visualiser: \url{https://koko-m.github.io/GoI-Visualiser/}}.




\section{A Term Calculus with Sub-Machine Semantics}
\label{sec:TermCalculus}

We use an untyped term calculus that accommodates three
evaluation strategies of the lambda-calculus, by dedicated constructors for function application: namely,
$\appLazy{}{}$ (call-by-need), $\appLR{}{}$ (left-to-right
call-by-value) and $\appRL{}{}$ (right-to-left call-by-value).
The term calculus uses all strategies so that we do not have to present three 
almost identical calculi. But we are not interested in their
interaction, but in each strategy separately.
In the rest of the paper, we therefore assume that each term contains
function applications of a single strategy.
As shown in the top of Fig.~\ref{fig:MockMachine},
the calculus accommodates explicit substitutions $[x \leftarrow u]$.
A term with no explicit substitutions is said to be ``pure.''

The sub-machine semantics is used to establish the soundness of the graph-rewriting
abstract machine.
It is an adaptation of Sinot's lambda-term rewriting system~\cite{Sinot05,Sinot06}, 
used to analyse a token-guided rewriting system for interaction nets.
It imitates an abstract machine by explicitly searching
for a redex and decomposing the meta-level substitution into on-demand
linear substitution, 
also resembling a storeless abstract machine (e.g.~\cite[Fig.~8]{DanvyMMZ12}).
However the semantics is still too ``abstract'' as an abstract
machine, in the sense that it works modulo alpha-equivalence to
avoid variable captures.

Fig.~\ref{fig:MockMachine} defines the sub-machine semantics of our
calculus.
It is given by labelled relations between \emph{enriched} terms
$\plug{E}{\W{t}}$.
In an enriched term $\plug{E}{\W{t}}$, a sub-term $t$ is not plugged directly into the evaluation 
context, but into a ``window'' $\W{\cdot}$ which makes it syntactically 
obvious where the reduction context is situated. 
Forgetting the window turns an enriched term into an ordinary term.
Basic rules $\mapsto$ are labelled with $\beta$, $\sigma$ or
$\epsilon$.
The basic rules~(2), (5) and~(8), labelled with $\beta$, apply
beta-reduction and delay substitution of a bound variable.
Substitution is done one by one, and on demand, by the basic rule (10)
with label~$\sigma$.
Each application of the basic rule~(10) replaces exactly one bound
variable with a value, and keeps a copy of the value for later use.
All  other basic rules, with label $\epsilon$, search for a
redex by moving the window without changing the underlying term.
Finally, reduction is defined by congruence of basic rules with
respect to evaluation contexts, and labelled accordingly.
Any basic rules and reductions are indeed between enriched terms,
because the window $\W{\cdot}$ is never duplicated or discarded.
\begin{figure}[t]
\begin{align*}
 t, u &::= x \mid \abs{x}{}{t}
 \mid \appLazy{t}{u} \mid \appLR{t}{u} \mid \appRL{t}{u}
 \mid t[x \leftarrow u],
 \quad v ::= \abs{x}{}{t}
 \tag{terms, values} \\
 A &::= \emptyCtxt \mid A[x \leftarrow t]
 \tag{answer contexts} \\
 E &::= \emptyCtxt \mid \appLazy{E}{t}
 \mid \appLR{E}{t} \mid \appLR{\plug{A}{v}}{E}
 \mid \appRL{t}{E} \mid \appRL{E}{\plug{A}{v}}
 \mid E[x \leftarrow t] \mid \plug{E}{x}[x \leftarrow E]
 \tag{evaluation contexts}
\end{align*}
 Basic rules $\mapsto_\beta$, $\mapsto_\sigma$ and $\mapsto_\epsilon$:
 \vspace{-2ex} \\
 \begin{minipage}{.5\hsize}
  \begin{align}
  \W{\appLazy{t}{u}} &\mapsto_\epsilon \appLazy{\W{t}}{u} \\
  \appLazy{\plug{A}{\W{\abs{x}{}{t}}}}{u}
  &\mapsto_\beta \plug{A}{\W{t}[x \leftarrow u]} \\
  \W{\appLR{t}{u}} &\mapsto_\epsilon \appLR{\W{t}}{u} \\
  \appLR{\plug{A}{\W{\abs{x}{}{t}}}}{u}
  &\mapsto_\epsilon \appLR{\plug{A}{\abs{x}{}{t}}}{\W{u}} \\
  \appLR{\plug{A}{\abs{x}{}{t}}}{\plug{A'}{\W{v}}}
  &\mapsto_\beta \plug{A}{\W{t}[x \leftarrow \plug{A'}{v}]}    
  \end{align}
 \end{minipage}
 \begin{minipage}{.5\hsize}
  \begin{align}
  \W{\appRL{t}{u}} &\mapsto_\epsilon \appRL{t}{\W{u}} \\
  \appRL{t}{\plug{A}{\W{v}}}
  &\mapsto_\epsilon \appRL{\W{t}}{\plug{A}{v}} \\
  \appRL{\plug{A}{\W{\abs{x}{}{t}}}}{\plug{A'}{v}}
  &\mapsto_\beta \plug{A}{\W{t}[x \leftarrow \plug{A'}{v}]} \\
  \plug{E}{\W{x}}[x \leftarrow \plug{A}{u}]
  &\mapsto_\epsilon \plug{E}{x}[x \leftarrow \plug{A}{\W{u}}] \\
  \plug{E}{x}[x \leftarrow \plug{A}{\W{v}}] &\mapsto_\sigma
  \plug{A}{\plug{E}{\W{v}}[x \leftarrow v]}
  \end{align}
 \end{minipage} \vspace{1ex}

 Reductions $\multimap_\beta$, $\multimap_\sigma$ and
 $\multimap_\epsilon$: \qquad
 \parbox[c]{4cm}{
  $\infer[(\chi \in \{ \beta,\sigma,\epsilon \})]
  {\plug{E}{\tilde{t}} \multimap_\chi \plug{E}{\tilde{u}}}
  {\tilde{t} \mapsto_\chi \tilde{u}}$
 }
 \caption{''Sub-Machine'' Operational Semantics}
 \label{fig:MockMachine}
\end{figure}

An \emph{evaluation} of a pure term $t$ (i.e.\ a term with no explicit
substitution) is a sequence of reductions starting from
$\plug{}{\W{t}}$, which is simply $\W{t}$.
In any evaluation, a sub-term in the window $\W{\cdot}$ is always pure.

\section{Token-Guided Graph-Rewriting Machine}\label{sec:Trans}


A graph is given by a set of nodes and a set of directed edges.
Nodes are classified into \emph{proper} nodes and \emph{link} nodes.
Each edge is directed, and at least one of its two endpoints is a link
node.
An \emph{interface} of a graph is given by two sets of link nodes,
namely \emph{input} and \emph{output}. 
Each link node is a source of at most one edge, and a target of at
most one edge.
Input links are the only links that are not a target of any edge, and
output links are the only ones that are not a source of any edge.
When a graph $G$ has $n$ input link nodes and $m$ output link nodes,
we sometimes write $G(n,m)$ to emphasise its interface.
If a graph has exactly one input, we refer to the input link node as
``root.''

The idea of using link nodes, as distinguished from proper nodes, 
comes from a graphical formalisation of string
diagrams~\cite{KissingerPhD}.
String diagrams consist of ``boxes'' that are connected to each other
by ``wires.''
In the formalisation, boxes are modelled by ``box-vertices''
(corresponding to proper nodes in our case), and wires are modelled by
consecutive edges connected via ``wire-vertices''
(corresponding to link nodes in our case).
The segmentation of wires into edges can introduce an arbitrary number
of consecutive link nodes, however these consecutive link nodes are
identified by the notion of ``wire homeomorphism.''
We will later discuss these consecutive link nodes, from the
perspective of the graph-rewriting machine.
From now on we simply call a proper node ``node,'' and a link node
``link.''

In drawing graphs, we follow the  convention that
input links are placed at the bottom and output links are at the top,
and links are usually not drawn explicitly.
The latter point means that edges are simply drawn from a node to a node,
with intermediate links omitted.
In particular if an edge is connected to an interface link, the edge
is drawn as an open edge missing an endpoint.
Additionally, we use a bold-stroke edge/node to represent a bunch of
parallel edges/nodes.

Nodes are labelled, and a node with a label $X$ is called an
``$X$-node.''
We use two sorts of labels.
One sort corresponds to the constructors of the calculus presented in
Sec.~\ref{sec:TermCalculus}, namely $\lambda$
(abstraction), $\appLazy{}{}$ (call-by-need application), $\appLR{}{}$
(left-to-right call-by-value application) and $\appRL{}{}$
(right-to-left call-by-value application).
These three application nodes are the novelty of this work.
The token, travelling in a graph, reacts to these nodes in
different ways, and hence implements different evaluation orders.
We believe that this is a more extensible way to accommodate different
evaluation orders, than to let the token react to the same node in
different ways depending on situation.
The other sort consists of $\oc$, $\wn$, $\lb{D}$ and $\lb{C}_n$ for
any natural number $n$, used in the management of copying sub-graphs.
This sort is inspired by proof nets of the multiplicative and
exponential fragment of linear logic~\cite{Girard87LL}, and
$\lb{C}_n$-nodes generalise the standard binary contraction and
incorporate  weakening.

\begin{wrapfigure}[6]{r}{.5\hsize}
 \centering
 \vspace{-2.5ex}
 \includegraphics[scale=0.8]{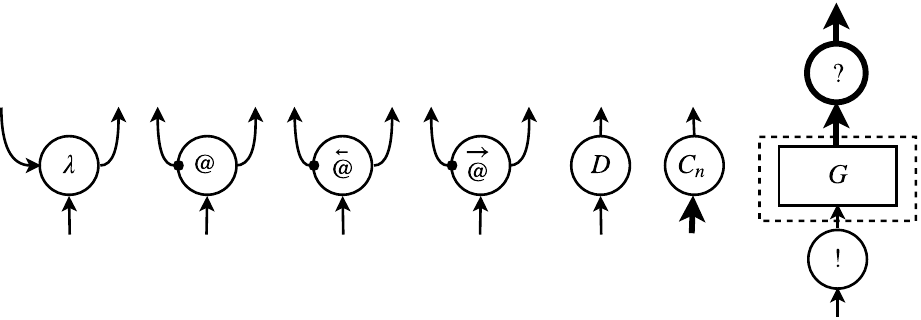}
 \vspace{-3ex}
 \caption{Graph Nodes}\label{fig:connection}
\end{wrapfigure}
The number of input/output and incoming/outgoing edges for a node is determined by the label, as 
indicated in Fig.~\ref{fig:connection}.
We distinguish two outputs of an application node
($\appLazy{}{}$, $\appLR{}{}$ or $\appRL{}{}$), calling one
``composition output'' and the other ``argument output''
(cf.~\cite{AccattoliG09}).
A bullet $\bullet$ in the figure specifies a function output.
The dashed box indicates a sub-graph $G(1,m)$
(``$\oc$-box'') that is connected to one $\oc$-node (``principal
door'') and $m$ $\wn$-nodes (``auxiliary doors'').
This $\oc$-box structure, taken from  proof nets, aids duplication of
sub-graphs by specifying those that can be copied\footnote{
Our formalisation of graphs is based on the view of proof nets as
string diagrams, and hence of $\oc$-boxes as functorial
boxes~\cite{Mellies06}. This allows dangling edges to be properly
modelled by link nodes~\cite{KissingerPhD},
which should not be confused with the terminology ``link'' of proof
nets.
}.

We define a graph-rewriting abstract machine as a labelled transition
system between \emph{graph states}.
\begin{definition}[Graph states]
 A \emph{graph state} $((G(1,0),e),\delta)$ is formed of a graph
 $G(1,0)$ with its distinguished link $e$, and token data
 $\delta = (d,f,S,B)$ that consists of:
   a \emph{direction} defined by
	$d ::= \up \mid \dn$,
   a \emph{rewrite flag} defined by
	$f ::= \square \mid \lambda \mid \oc$,
  a \emph{computation stack} defined by
	$S ::= \square \mid \star:S \mid \lambda:S \mid @:S$, and
  a \emph{box stack} defined by
	$B ::= \square \mid \star:B \mid \oc:B \mid
	\diamond:B \mid e':B$,
	where $e'$ is any link of the graph $G$.
\end{definition}
\noindent
The distinguished link $e$  is called the
``position'' of the token.
The token reacts to a node in a graph using its data, which determines its
path.
Given a graph $G$ with root $e_0$,
the \emph{initial} state $\Init(G)$ on it is given by
$((G,e_0),(\up,\square,\square,\star:\square))$, and the \emph{final}
state $\Final(G)$ on it is given by
$((G,e_0),(\dn,\square,\square,\oc:\square))$.
An \emph{execution} on a graph $G$ is a sequence of transitions
starting from the initial state $\Init(G)$.

Each transition
$((G,e),\delta) \to_\chi ((G',e'),\delta')$ between graph states is
labelled by either $\beta$, $\sigma$ or $\epsilon$.
Transitions are deterministic, and classified into \emph{pass}
transitions that search for redexes and trigger rewriting, and 
\emph{rewrite} transitions that actually rewrite a graph as soon as a
redex is found.

A pass transition
$((G,e),(d,\square,S,B)) \to_\epsilon ((G,e'),(d',f',S',B'))$,
always labelled with $\epsilon$, applies to a state whose rewrite flag
is $\square$.
It simply moves the token over one node, and updates its data by
modifying the top elements of stacks, while keeping an underlying
graph unchanged.
When the token passes a $\lambda$-node or a $\oc$-node, a rewrite
flag is changed to $\lambda$ or $\oc$, which triggers rewrite
transitions.
Fig.~\ref{fig:PassTransitions} defines pass transitions, by showing
only the relevant node for each transition.
The position of the token is drawn as a black triangle, pointing
towards the direction of the token.
In the figure, $X \neq \star$, and $n$ is a natural number.
The pass transition over a $\lb{C}_{n+1}$-node pushes the old position
$e$, a link node, to a box stack.
\begin{figure}[t]
 \centering
 \includegraphics[scale=0.8]{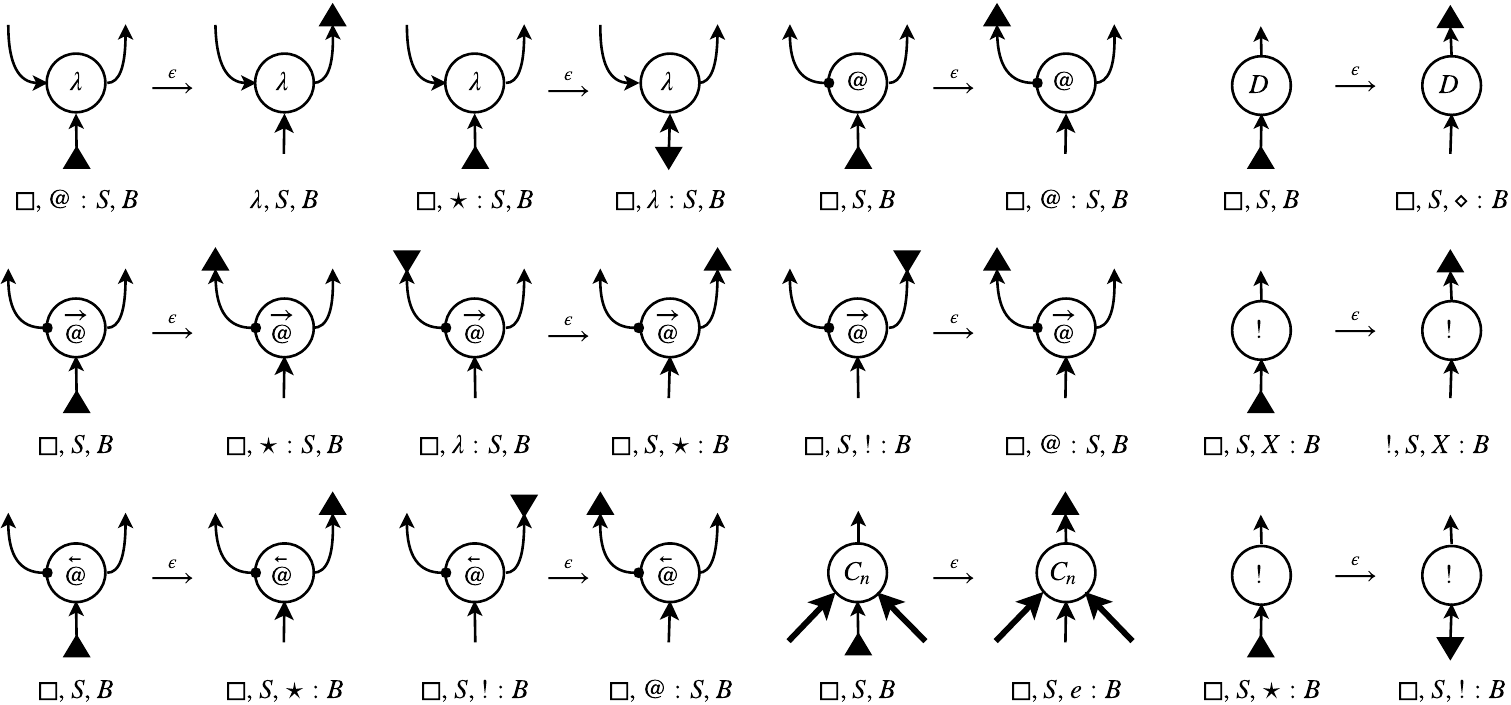}
 \caption{Pass Transitions}
 \label{fig:PassTransitions}
\end{figure}

The way the token reacts to application nodes ($\appLazy{}{}$,
$\appLR{}{}$ and $\appRL{}{}$) corresponds to the way the window
$\W{\cdot}$ moves in evaluating these function applications in the
sub-machine semantics (Fig.~\ref{fig:MockMachine}).
When the token moves on to the composition output of an application
node, the top element of a computational stack is either $@$ or
$\star$.
The element $\star$ makes the token return  from a $\lambda$-node,
which corresponds to reducing the function part of application to a
value (i.e.\ abstraction).
The element $@$ lets the token proceed at a $\lambda$-node, raises the
rewrite flag $\lambda$, and hence triggers a rewrite transition that
corresponds to beta-reduction.
The call-by-value application nodes ($\appLR{}{}$ and $\appRL{}{}$)
send the token to their argument output,  pushing the element
$\star$ to a box stack.
This makes the token bounce at a $\oc$-node and return 
to the application node, which corresponds to evaluating the argument
part of function application to a value.
Finally, pass transitions through $\lb{D}$-nodes, $\lb{C}_n$-nodes and
$\oc$-nodes prepare copying of values, and eventually raise the
rewrite flag $\oc$ that triggers on-demand duplication.

A rewrite transition
$((G,e),(d,f,S,B)) \to_\chi ((G',e'),(d',f',S,B'))$,
labelled with $\chi \in \{ \beta,\sigma,\epsilon \}$, applies to a
state whose rewrite flag is either $\lambda$ or $\oc$.
It changes a specific sub-graph while keeping its interface, changes
the position accordingly, and pops an element from a box stack.
Fig.~\ref{fig:RewriteTransitions} defines rewrite transitions by
showing a sub-graph (``redex'') to be rewritten.
Before we go through each rewrite transition, we note that rewrite
transitions are not exhaustive in general, as a graph may not match a
redex even though a rewrite flag is raised.
However we will see that there is no failure of transitions in
implementing the term calculus.
\begin{figure}[t]
 \centering
 \includegraphics[scale=0.8]{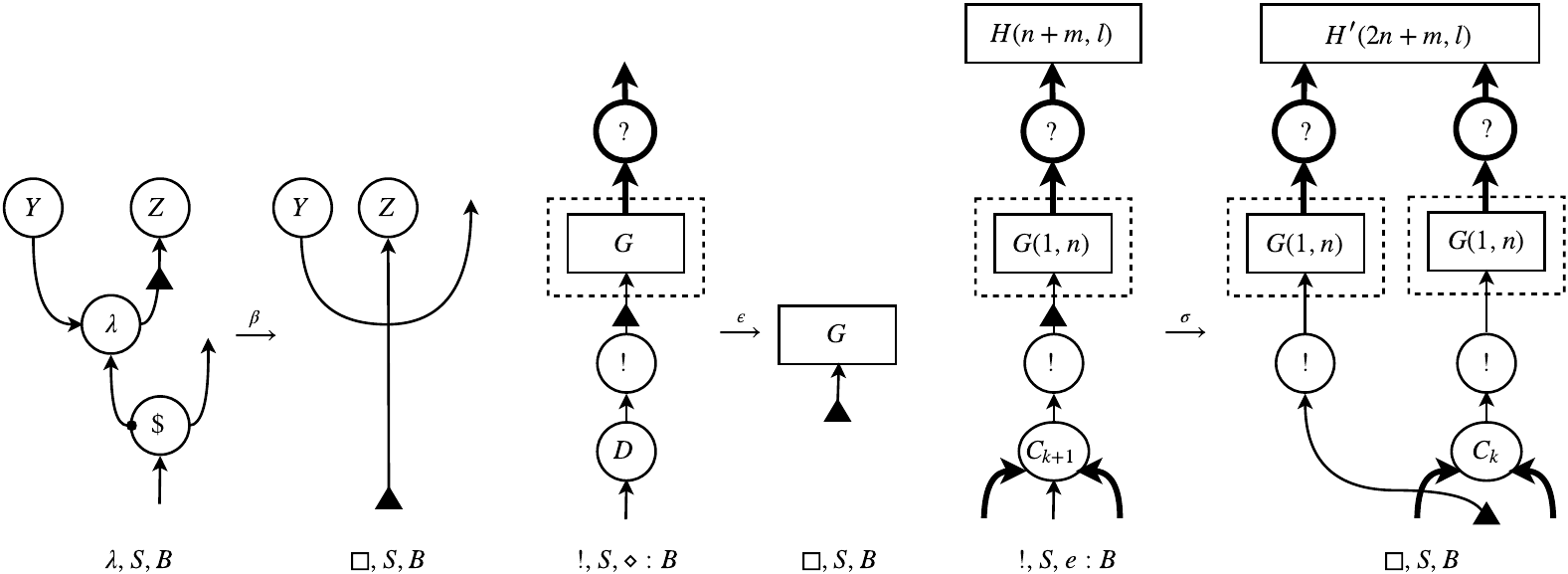}
 \caption{Rewrite Transitions}
 \label{fig:RewriteTransitions}
\end{figure}

The first rewrite transition in Fig.~\ref{fig:RewriteTransitions},
with  label $\beta$, occurs when a rewrite flag is $\lambda$.
It implements beta-reduction by eliminating a pair of an abstraction
node ($\lambda$) and an application node
($\$ \in \{ \appLazy{}{},\appLR{}{},\appRL{}{} \}$ in the figure).
Outputs of the $\lambda$-node are required to be connected to
arbitrary nodes (labelled with $Y$ and $Z$ in the figure), so that
edges between links are not introduced.
The other rewrite transitions are for the rewrite flag $\oc$, and
they together realise the copying process of a sub-graph (namely a
$\oc$-box).
The second rewrite transition in Fig.~\ref{fig:RewriteTransitions},
labelled with $\epsilon$, finishes off each copying process by
eliminating all doors of the $\oc$-box $G$.
It replaces the interface of $G$ with output links of the auxiliary
doors and the input link of the $\lb{D}$-node, which is the new
position of the token, and pops the top element
$\diamond$ of a box stack.
Again, no edge between links are introduced.

\begin{wrapfigure}[11]{r}{.5\hsize}
 \centering
 \vspace{-1ex}
 \includegraphics[scale=0.8]{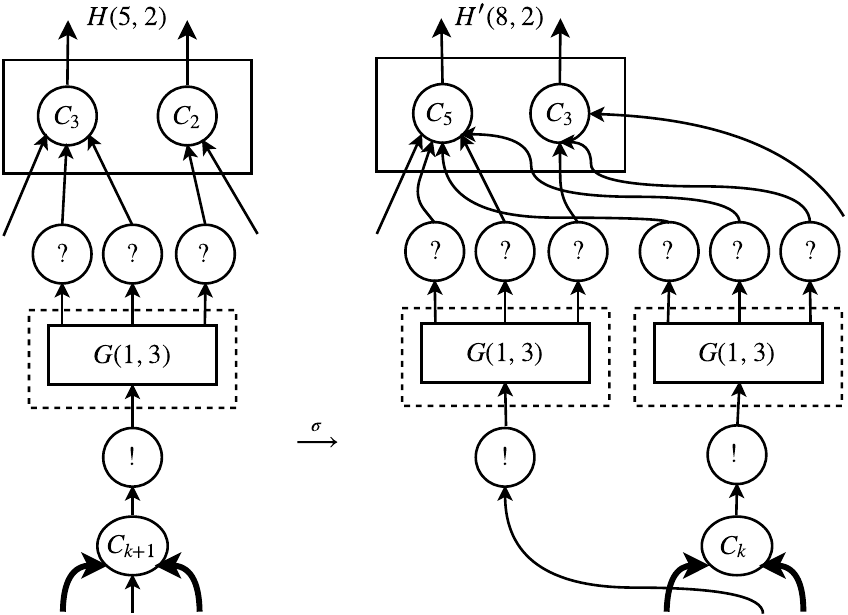}
\end{wrapfigure}
The last rewrite transition in the figure, with label $\sigma$,
actually copies a $\oc$-box.
It requires the top element $e$ of the old box stack to be one of
input links of the $\lb{C}_{k+1}$-node (where $k$ is a natural
number).
The link $e$ is popped from the box stack and becomes the new position
of the token, and the $\lb{C}_{k+1}$-node becomes a
$\lb{C}_k$-node by keeping all the inputs except for the link $e$.
The sub-graph $H(n+m,l)$ consists of $l$ parallel $\lb{C}$-nodes that
altogether have $n+m$ inputs.
Among these inputs, $n$ are connected to auxiliary doors of the
$\oc$-box $G(1,n)$, and $m$ are connected to nodes that are not in the
redex.
The sub-graph $H(n+m,l)$ is turned into $H'(2n+m,l)$ by introducing $n$
inputs to these $\lb{C}$-nodes as follows: if an auxiliary
door of the $\oc$-box $G$ is connected to a $\lb{C}$-node in $H$, two
copies of the auxiliary door are both connected to the corresponding
$\lb{C}$-node in $H'$.
Therefore the two sub-graphs consist of the same number $l$ of
$\lb{C}$-nodes, whose indegrees are possibly increased.
The $m$ inputs, connected to nodes outside a redex, are kept
unchanged.
For example, copying a graph $G(1,3)$ for $H(5,2)$
will give an $H'(8,2)$ as shown above.

All pass and rewrite transitions are well-defined. The following
``sub-graph'' property is essential in time-cost analysis, because it
bounds the size of duplicable sub-graphs (i.e.\ $\oc$-boxes) in an
execution.
\begin{lemma}[Sub-graph property]
 \label{lem:SubGraph}
 For any execution $\Init(G) \to^* \Final((H,e),\delta)$, each
 $\oc$-box of the graph $H$ appears as a sub-graph of the initial
 graph $G$.
\end{lemma}
\begin{proof}
 Rewrite transitions can only copy or discard a $\oc$-box, and cannot
 introduce, expand or reduce a single $\oc$-box. Therefore, any
 $\oc$-box of $H$ has to be already a $\oc$-box of the initial graph
 $G$.
\end{proof}


When a graph has an edge between links, the token is just passed along.
With this pass transition over a link at hand, the equivalence relation between
graphs that identifies consecutive links with a single
link---so-called ``wire homeomorphism''~\cite{KissingerPhD}---lifts to
a weak bisimulation between graph states.
Therefore, behaviourally, we can safely ignore consecutive links.
From the perspective of time-cost analysis, we benefit from the fact
that rewrite transitions are designed not to introduce any edge between
links.
This means, by assuming that an execution starts with a graph with no
consecutive links, we can analyse time cost of the execution without
caring the extra pass transition over a link.

\section{Implementation of Evaluation Strategies}

The implementation of the term calculus, by means of the dynamic GoI,
starts with translating (enriched) terms into graphs.
The definition of the translation uses multisets of variables, to
track how many times each variable occurs in a term. We assume that terms
are alpha-converted in a form in which all binders introduce distinct
variables.
\begin{notation*}[Multiset]
 The empty multiset is denoted by $\emptyset$, and the sum of two
 multisets $M$ and $M'$ is denoted by $M+M'$.
 We write $x \in^k M$ if the multiplicity of $x$ in a multiset $M$ is
 $k$.
 Removing \emph{all} $x$ from a multiset $M$ yields the multiset
 $M \backslash x$, e.g.\ ${[x,x,y] \backslash x} = [y]$.
 We abuse the notation and refer to a multiset $[x,\ldots,x]$ of a
 finite number of $x$'s, simply as $x$.
\end{notation*}
\begin{definition}[Free variables]
 The map $\FV$ of terms to multisets of variables is inductively
 defined by:
 \begin{align*}
  \FV(x) &:= [x], &
  \FV(\abs{x}{}{t}) &:= \FV(t) \backslash x, \\ 
  \FV(\appGen{t}{u}) &:= \FV(t) + \FV(u), &
  \FV(t[x \leftarrow u]) &:= (\FV(t) \backslash x) + \FV(u).
  \tag{$\appGen{}{} \in \{ \appLazy{}{},\appLR{}{},\appRL{}{} \}$}
 \end{align*}
For a multiset $M$ of variables,
the map $\FV_M$ of
 evaluation contexts to multisets of variables is
 defined by:
 \begin{align*}
  \FV_M(\emptyCtxt) &:= M, &
  \FV_M(\appRL{t}{E}) &:= \FV(t) + \FV_M(E), \\
  \FV_M(\appLazy{E}{t}) &:= \FV_M(E) + \FV(t), &
  \FV_M(\appRL{E}{\plug{A}{v}}) &:= \FV_M(E) + \FV(\plug{A}{v}), \\
  \FV_M(\appLR{E}{t}) &:= \FV_M(E) + \FV(t), &
  \FV_M(E[x \leftarrow t]) &:=
  (\FV_M(E) \backslash x) + \FV(t), \\
  \FV_M(\appLR{\plug{A}{v}}{E}) &:= \FV(\plug{A}{v}) + \FV_M(E), &
  \FV_M(\plug{E'}{x}[x \leftarrow E]) &:=
  (\FV(\plug{E'}{x}) \backslash x) + \FV_M(E).
 \end{align*}
\end{definition}

A term $t$ is said be \emph{closed} if $\FV(t) = \emptyset$.
Consequences of the above definition are the following equations,
where $M'$ is not captured in $E$.
\begin{align*}
 \FV(\plug{E}{t}) = \FV_{\FV(t)}(E), \qquad
 \FV_M(\plug{E}{E'}) = \FV_{\FV_M(E')}(E), \qquad
 \FV_{M+M'}(E) = \FV_M(E) + M'.
\end{align*}

\begin{wrapfigure}[4]{r}{.3\hsize}
\centering
	\includegraphics[scale=0.8]{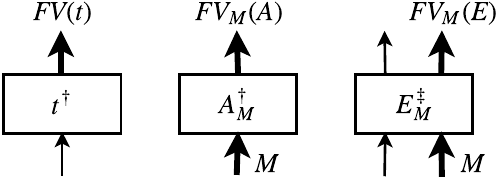}
\end{wrapfigure}
We give translations of terms, answer contexts, and evaluation
contexts separately.
Fig.~\ref{fig:TermAnsCtxtTranslation} and
Fig.~\ref{fig:EvalCtxtTranslation} define
two mutually recursive translations $(\cdot)^\dag$ and $(\cdot)^\ddag$, the 
first one for terms and answer contexts, and the second one for
evaluation contexts.
In the figures,
${\appGen{}{}} \in \{ \appLazy{}{},\appLR{}{},\appRL{}{} \}$,
and $m$ is the multiplicity of $x$.
The general form of the translations is as shown right.

The annotation of bold-stroke edges means each edge of a bunch is
labelled with an element of the annotating multiset, in a one-to-one
manner.
In particular if a bold-stroke edge is annotated by a variable $x$,
all edges in the bunch are annotated by the variable $x$.
These annotations are only used to define the translations, and
are subsequently ignored during execution.

The translations are based on the so-called ``call-by-value''
translation of linear logic to intuitionistic logic (e.g.\
\cite{MaraistOTW99}).
Only the translation of abstraction can be accompanied by a
$\oc$-box, which captures the fact
that only values (i.e.\ abstractions) can be duplicated
(see the basic rule (10) in Fig.~\ref{fig:MockMachine}).
Note that only one $\lb{C}$-node is introduced for each bound
variable. This is vital to achieve constant cost in looking up a
variable, namely in realising the basic rule (9) in
Fig.~\ref{fig:MockMachine}.

\begin{figure}[t]
 \centering
 \includegraphics[scale=0.8]{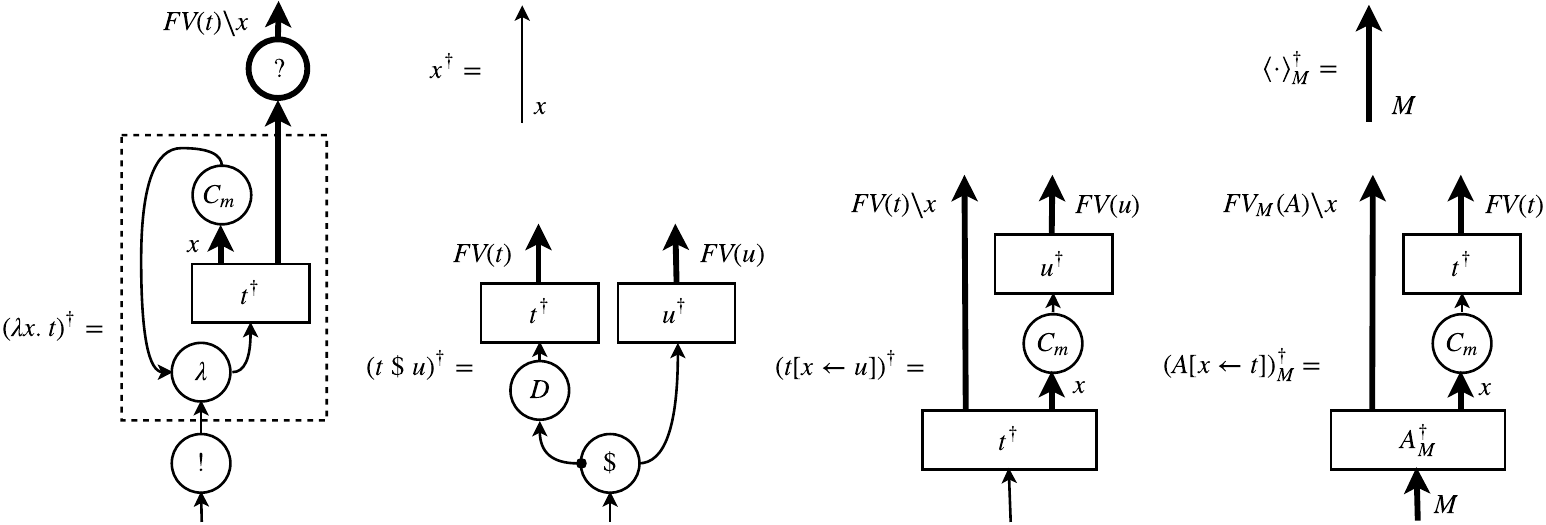}
 \caption{Inductive Translation of Terms and Answer Contexts}
 \label{fig:TermAnsCtxtTranslation}
\end{figure}
\begin{figure}[t]
 \centering
 \includegraphics[scale=0.8]{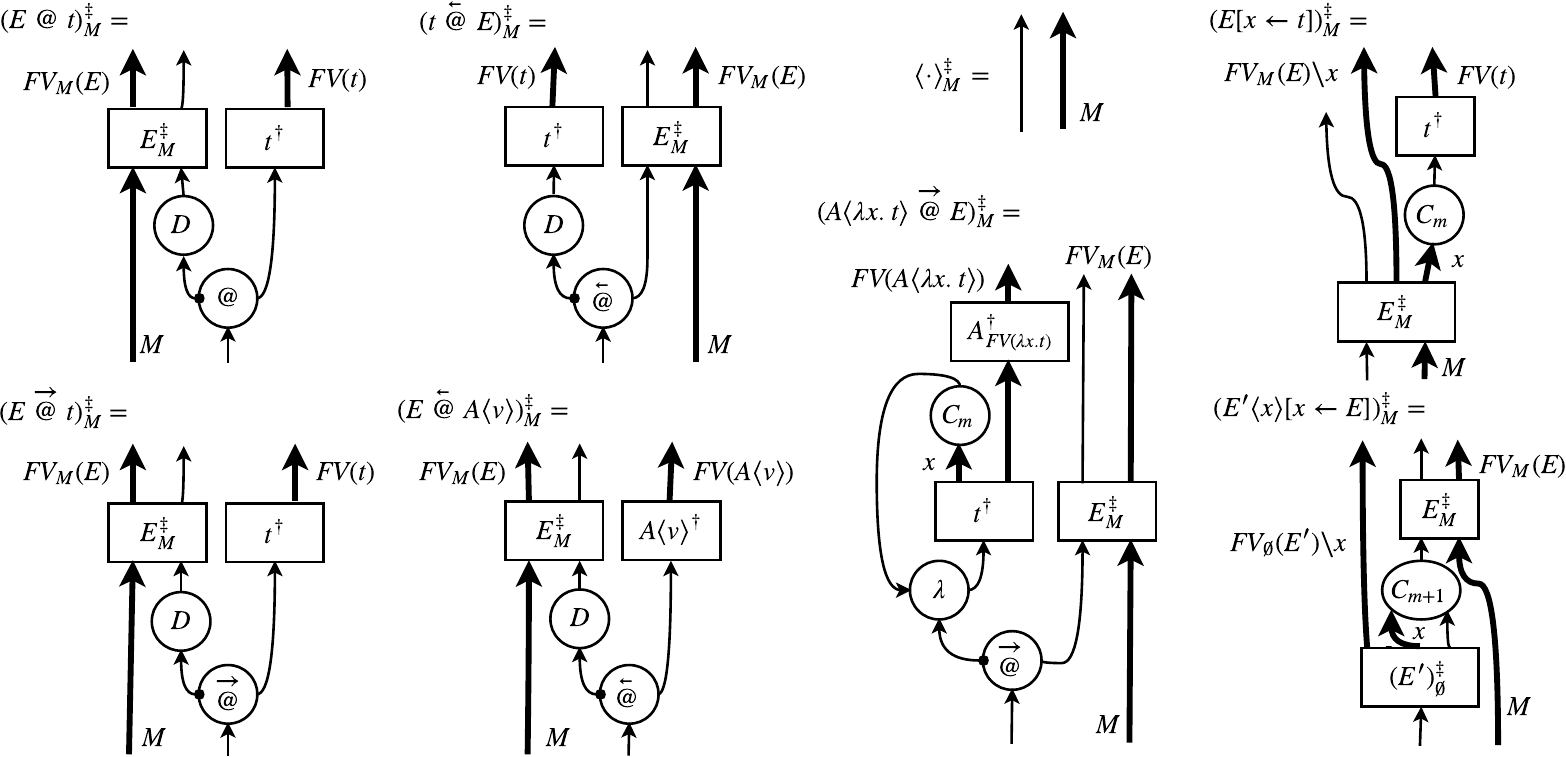}
 \caption{Inductive Translation of Evaluation Contexts}
 \label{fig:EvalCtxtTranslation}
\end{figure}

The two mutually recursive translations $(\cdot)^\dag$ and
$(\cdot)^\ddag$ are related by the following decompositions, which can
be checked by straightforward induction.
In the third decomposition, $M'$ is not captured in $E$.
\begin{center}
 \includegraphics[scale=0.8]{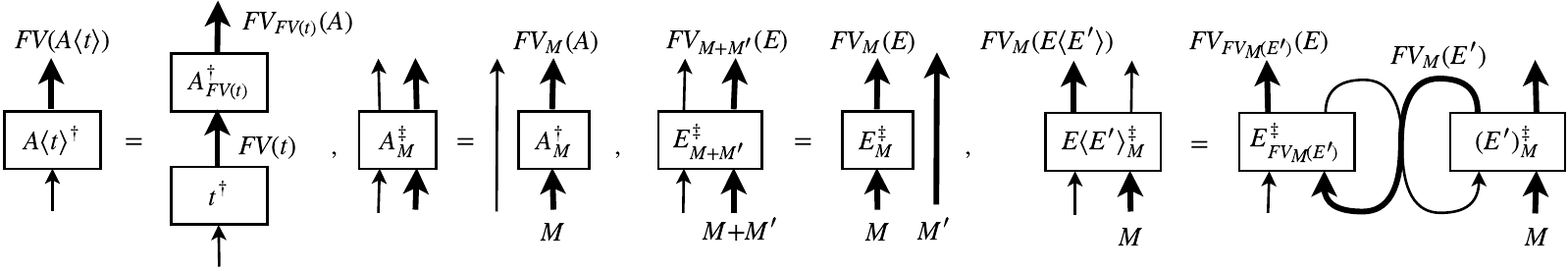}
\end{center}
Note that the decomposition property like the fourth one does not hold
for $\plug{E}{t}$
in general, because a translation
$(\appLR{\plug{A}{\abs{x}{}{t}}}{E})^\ddag_M$ lacks a $\oc$-box
structure, compared to a translation
$(\appLR{\plug{A}{\abs{x}{}{t}}}{u})^\dag$.

The inductive translations lift to
a binary relation between closed enriched terms and graph states.
\begin{definition}[Binary relation $\preceq$]
 The binary relation $\preceq$ is defined by
 $\plug{E}{\W{t}} \preceq
 ((E^\ddag \circ t^\dag,e), (\up,\square,S,B))$, where:
 (i) $\plug{E}{\W{t}}$ is a closed enriched term, and
 $(E^\ddag \circ t^\dag,e)$ is given by
 \parbox[c]{3cm}{\includegraphics[scale=0.8]{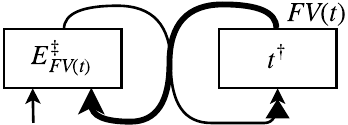}} with no
 edges between links,
 and (ii) there is an execution
 $\Init(E^\ddag \circ t^\dag) \to^*
 ((E^\ddag \circ t^\dag,e),(\up,\square,S,B))$
 such that the position $e$ appears only in the last state of the
 sequence.
\end{definition}
A special case is $\W{t} \preceq \Init(t^\dag)$, which relates the
starting points of an evaluation and an execution.
We require the graph $E^\ddag \circ t^\dag$ to have no edges between
links, which is based on the discussion at the end of
Sec.~\ref{sec:Trans} and essential for time-cost analysis.
Although the definition of the translations relies on edges between links
(e.g.\ the translation $x^\dag$),
we can safely replace any consecutive links in the composition of
translations $E^\ddag$ and $t^\dag$ with a single link, and yield the
graph $E^\ddag \circ t^\dag$ with no consecutive links.

The binary relation $\preceq$ gives a weak simulation of the
sub-machine semantics by the graph-rewriting machine.
The weakness, i.e.\ the extra transitions compared with reductions,
comes from the locality of pass transitions and the bureaucracy of
managing $\oc$-boxes.
\begin{theorem}[Weak simulation with global bound]
 \label{thm:WeakSimulation}
 \noindent
 \begin{enumerate}
  \item If $\plug{E}{\W{t}} \multimap_\chi \plug{E'}{\W{t'}}$ and
	$\plug{E}{\W{t}} \preceq ((E^\ddag \circ t^\dag,e),\delta)$
	hold, then there exists a number $n \leq 3$ and	a graph state
	$(((E')^\ddag \circ (t')^\dag,e'),\delta')$ such that
	$((E^\ddag \circ t^\dag,e),\delta) \to_\epsilon^n \to_\chi
	(((E')^\ddag \circ (t')^\dag,e'),\delta')$
	and $\plug{E'}{\W{t'}} \preceq
	(((E')^\ddag \circ (t')^\dag,e'),\delta')$.
  \item If $\plug{A}{\W{v}} \preceq
	((A^\ddag \circ v^\dag,e),\delta)$ holds,
	then the graph state $((A^\ddag \circ v^\dag,e),\delta)$ is
	initial, from which only the transition
	$\Init(A^\ddag \circ v^\dag) \to_\epsilon
	\Final(A^\ddag \circ v^\dag)$
	is possible.
 \end{enumerate}
\end{theorem}
\begin{proof}[Proof outline]
 The second half of the theorem is straightforward. For the first
 half,
 Fig.~\ref{fig:Proof1} and Fig.~\ref{fig:Proof2}
 illustrate how the graph-rewriting machine
 simulates each reduction $\multimap$ of the sub-machine semantics.
 Annotations of edges are omitted.
 The figures altogether include ten sequences of translations $\to$,
 whose only first and last graph states are shown.
 Each sequence simulates a single reduction $\multimap$, and
 is preceded by a number (i.e.\ (1)) that corresponds to a basic rule 
 applied by the reduction (see Fig.~\ref{fig:MockMachine}).
 Some sequences involve equations that apply the four
 decomposition properties of the translations $(\cdot)^\dagger$ and
 $(\cdot)^\ddagger$, which are given earlier in this section.
 These equations rely on the fact that reductions with
 labels $\beta$ and $\sigma$ work modulo alpha-equivalence to avoid
 name captures.
 This means that (i) free variables of $u$ (resp.\
 $\plug{A'}{v}$) are never captured by $A$ in the reduction (2)
 (resp.\ (5) and (8)), (ii) the variable $x$ is never captured by $E$
 or $E'$, and (iii) free variables of $E$ are never captured by $A$.
 Especially in simulation of the reduction (9), the variable $x$ is
 not captured by the evaluation context $E'$, and therefore the first
 token position is in fact an input of the $\lb{C}_{m+1}$-node.
\end{proof}
\begin{figure}[p]
 \centering
 \includegraphics[scale=0.9]{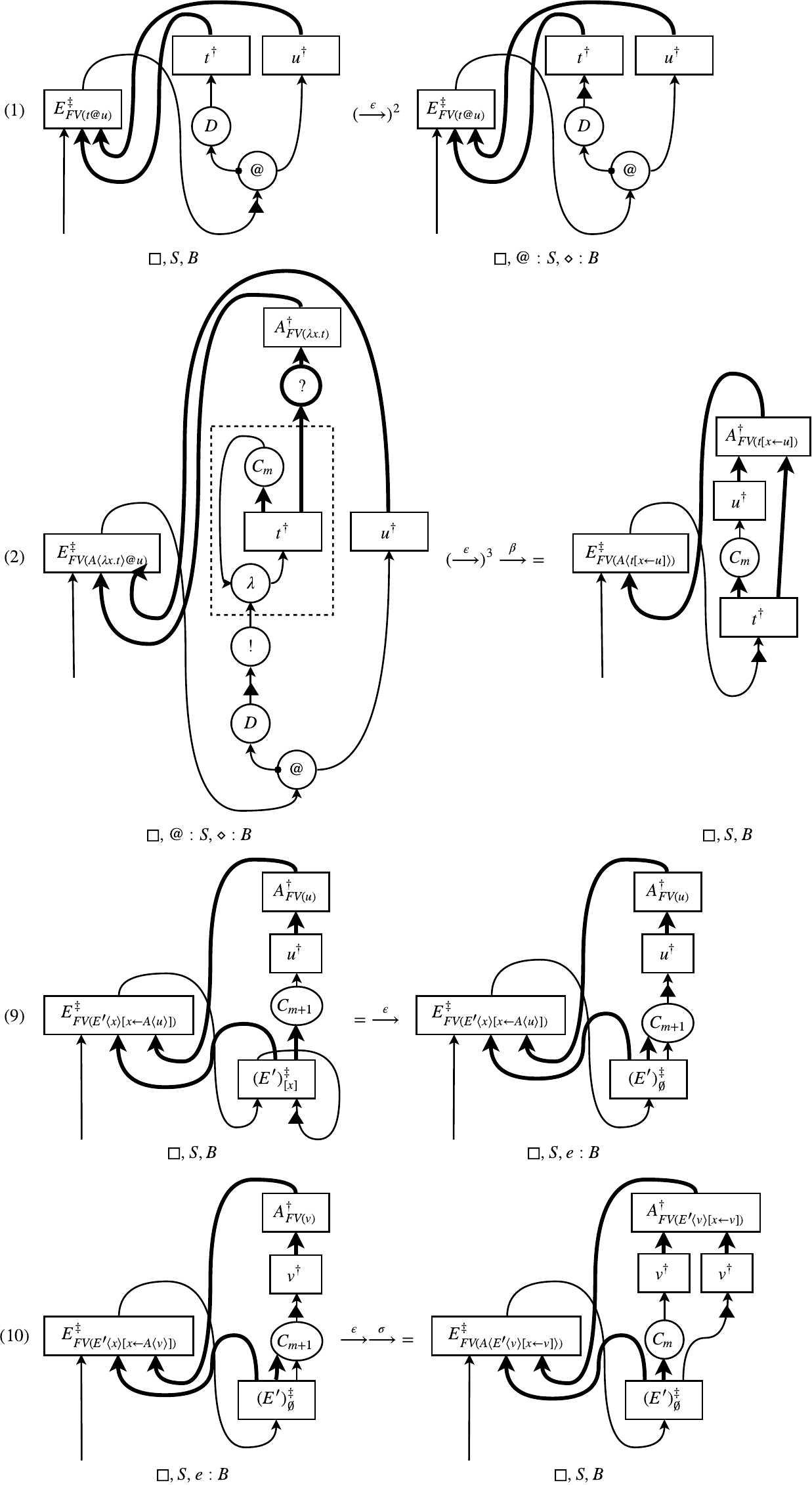}
 \caption{Illustration of Simulation: Call-by-Need Application and
 Explicit Substitutions}
 \label{fig:Proof1}
\end{figure}
\begin{figure}[p]
 \centering
 \includegraphics[scale=1.25]{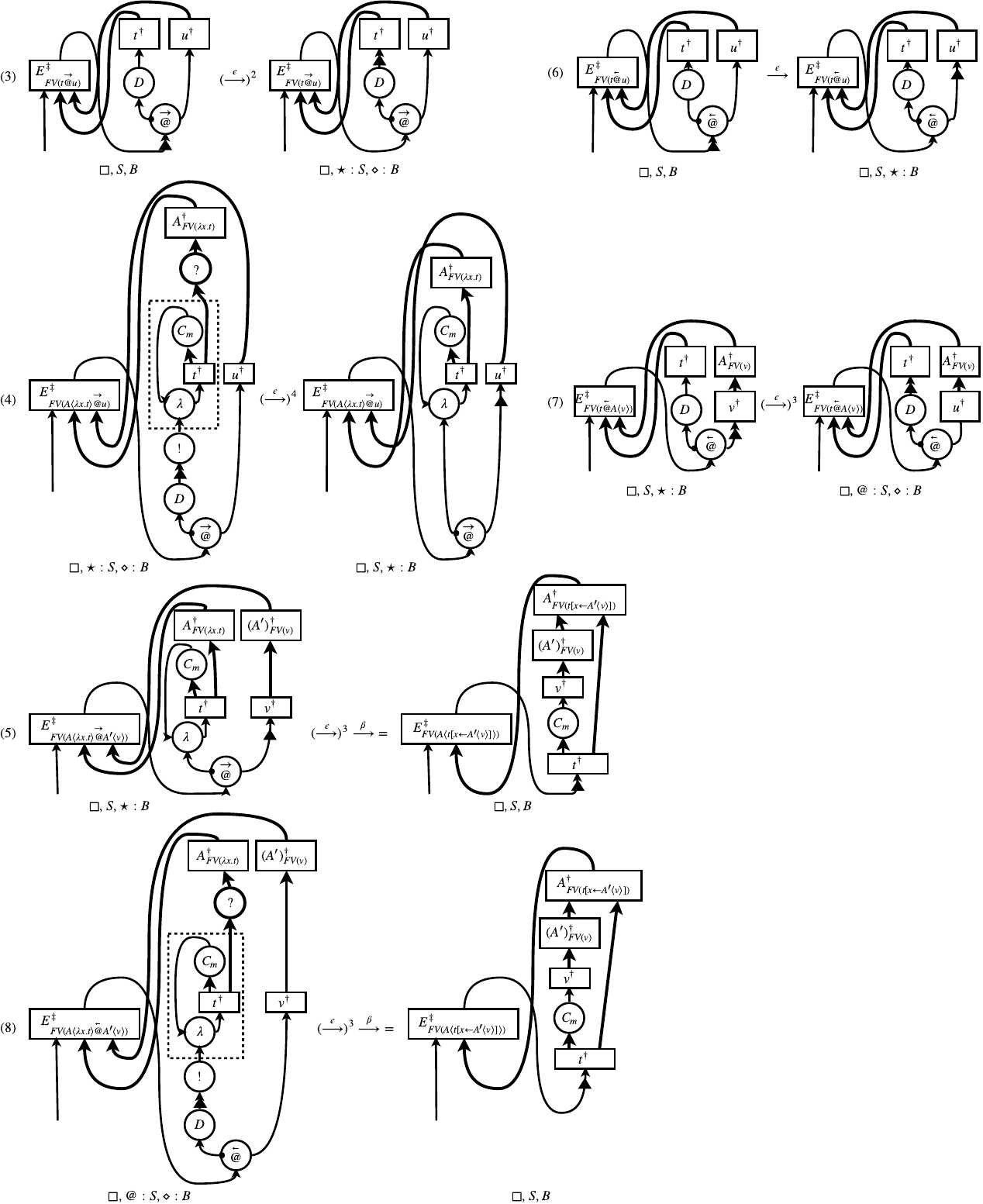}
 \caption{Illustration of Simulation: Left-to-Right and Right-to-Left Call-by-Value Applications}
 \label{fig:Proof2}
\end{figure}

We analyse how time-efficiently the token-guided graph-rewriting
machine implements evaluation strategies, following the methodology
developed by Accattoli et al.\
\cite{AccattoliBM14,AccattoliSC14,Accattoli16}.
The methodology tracks the number of beta-reduction steps
in an evaluation in three steps:
(I) bound the number of transitions required in implementing
evaluation strategies, (II) estimate time cost of each transition, and
(III) bound overall time cost of implementing evaluation strategies,
by multiplying the number of transitions with time cost for each
transition.

Given a pure term $t$, the time cost of an execution on the graph $t^\dag$
is estimated by means of: (i) the number
of reductions labelled with $\beta$ in the evaluation of the term $t$,
and (ii) the \emph{size} $|t|$ of the term $t$, inductively defined as:
$ |x| := 1$,
$ |\abs{x}{}{t}| := |t| + 1$, 
$ |\appLazy{t}{u}| = |\appLR{t}{u}| = |\appRL{t}{u}| := |t| + |u| + 1$,
$ |t[x \leftarrow u]| := |t| + |u| + 1$.

Given an evaluation $\ev$, the number of occurrences of a label $\chi$
is denoted by $|\ev|_\chi$.
The sub-machine semantics comes with the following quantitative
bounds.
\begin{proposition}
 \label{prop:MockMachineBounds}
 For any evaluation $\ev \colon \W{t} \to^* \plug{A}{\W{v}}$ that
 terminates, the number of reductions is bounded by
 $
  |\ev|_\sigma = \mathcal{O}(|\ev|_\beta)$ and $
  |\ev|_\epsilon = \mathcal{O}(|t|\cdot|\ev|_\beta).
 $
\end{proposition}
\begin{proof}[Proof outline]
 A term uses a single evaluation strategy,
   either call-by-need, left-to-right call-by-value, or
 right-to-left call-by-value.
 The proof is by developing the one-to-one correspondence between an
 evaluation by the sub-machine semantics and a ``derivation'' in the
 linear substitution calculus.
 This goes in the same way Accattoli et al.\
 analyse various abstract machines~\cite{AccattoliBM14}, especially
 the proof of the second equation~\cite[Thm.~11.3 \&
 Thm.~11.5]{AccattoliBM14}.
 The first equation is a direct application of the bounds about the
 linear substitution calculus~\cite[Cor.~1 \& Thm.~2]{AccattoliSC14}.
\end{proof}

We use the same notation $|\ex|_\chi$, as for an evaluation, to denote
the number of occurrences of each label $\chi$ in an execution $\ex$.
Additionally the number of rewrite transitions with the label
$\epsilon$ is denoted by $|\ex|_\mathit{\epsilon R}$.
The following proposition completes the first step of
the cost analysis.
\begin{proposition}[Soundness \& completeness, with number bounds]
 \label{prop:GraphRewriterBounds}
 For any pure closed term $t$,
 an evaluation $\ev \colon \W{t} \to^* \plug{A}{\W{v}}$
 terminates with the enriched term $\plug{A}{\W{v}}$ if and only if an
 execution
 $\ex \colon
 \Init(t^\dag) \to^* \Final(A^\ddag \circ v^\dag)$
 terminates with the graph $A^\ddag \circ v^\dag$.
 Moreover the number of transitions is bounded by
$  |\ex|_\beta = |\ev|_\beta$,
$  |\ex|_\sigma = \mathcal{O}(|\ev|_\beta)$,
$  |\ex|_\epsilon = \mathcal{O}(|t|\cdot|\ev|_\beta)$,
$  |\ex|_\mathit{\epsilon R} = \mathcal{O}(|\ev|_\beta)$.
\end{proposition}
\begin{proof}
 This proposition is a direct consequence of
 Thm.~\ref{thm:WeakSimulation} and Prop.~\ref{prop:MockMachineBounds},
 except for the last bound.
 The last bound of $|\ex|_\mathit{\epsilon R}$ follows from the fact
 that
 each rewrite transition labelled with $\beta$ is always preceded by
 one rewrite transition labelled with $\epsilon$.
\end{proof}

The next step in the cost analysis is to estimate the time cost of each transition.
We assume that graphs are implemented in the following way.
Each link is given by two pointers to its child and its parent, and
each node is given by its label and pointers to its outputs.
Abstraction nodes ($\lambda$) and
application nodes ($\appLazy{}{}$, $\appLR{}{}$ and
$\appRL{}{}$) have two pointers that are distinguished, and all
the other nodes have only one pointer to their unique output.
Additionally each $\oc$-node has pointers to inputs of its
associated $\wn$-nodes, to represent a $\oc$-box structure.
Accordingly, a position of the token is a pointer to a link, a
direction and a rewrite flag are two symbols, a computation stack is a
stack of symbols, and finally a box stack is a stack of symbols and
pointers to links.

Using these assumptions of implementation, we estimate time cost of
each transition.
All pass transitions have constant cost.
Each pass transition looks up one node and its outputs (that
are either one or two) next to the current position, and involves
a fixed number of elements of the token data.
Rewrite transitions with the label $\beta$ have constant cost,
as they change a constant number of nodes and links, and only a
rewrite flag of the token data.
Rewrite transitions with the label $\epsilon$ remove a
$\oc$-box structure, and hence have cost bounded
by the number of the auxiliary doors.
Finally, rewrite transitions with the label $\sigma$ copy a $\oc$-box
structure.
Copying cost is bounded by the size of the $\oc$-box, i.e.\ the
number of nodes and links in the $\oc$-box.
Updating cost of the sub-graph $H'$ (see
Fig.~\ref{fig:RewriteTransitions}) is bounded by the number of
auxiliary doors, that is less than the size of the copied
$\oc$-box.
The assumption about the implementation of graphs enables us to
conclude updating cost of the $\lb{C}$-node is constant.

With the results of the previous two steps, we can now give the
overall time cost of executions.
\begin{theorem}[Soundness \& completeness, with cost bounds]
 \label{thm:CostBounds}
 For any pure closed term $t$,
 an evaluation $\ev \colon \W{t} \to^* \plug{A}{\W{v}}$
 terminates with the enriched term $\plug{A}{\W{v}}$ if and only if an
 execution
 $\ex \colon
 \Init(t^\dag) \to^* \Final(A^\ddag \circ v^\dag)$
 terminates with the graph $A^\ddag \circ v^\dag$.
 The overall time cost of the execution $\ex$ is bounded by
 $\mathcal{O}(|t|\cdot|\ev|_\beta)$.
\end{theorem}
\begin{proof}
 Non-constant cost of rewrite transitions are either the number of
 auxiliary doors of a $\oc$-box or the size of a $\oc$-box.
 The former can be bounded by the latter, which is no more than the
 size of the initial graph $t^\dag$, by Lem.~\ref{lem:SubGraph}.
 The size of the initial graph $t^\dag$ can be bounded by the size
 $|t|$ of the initial term.
 Therefore any non-constant cost of each rewrite transition, in the
 execution $\ex$, can be also bounded by $|t|$.
 The overall time cost of rewrite transitions labelled with $\beta$ is
 $\mathcal{O}(|\ev|_\beta)$, and that of the other rewrite transitions
 and pass transitions is $\mathcal{O}(|t|\cdot|\ev|_\beta)$.
\end{proof}

Thm.~\ref{thm:CostBounds} classifies the graph-rewriting machine as
``efficient,'' by Accattoli's taxonomy~\cite[Def.~7.1]{Accattoli16} of
abstract machines.
The efficiency benefits from the graphical representation of
environments (i.e.\ explicit substitutions in our setting).
In particular the translations $(\cdot)^\dag$ and $(\cdot)^\ddag$ are
carefully designed to exclude any two sequentially-connected
$\lb{C}$-nodes, which yields the constant cost to look up a bound
variable and its associated computation in environments.

\section{Related Work and Conclusions}

In an abstract machine of any functional programming language,
computations assigned to variables have to be stored for later use.
Potentially multiple, conflicting, computations can be assigned to a
single variable, primarily because of
multiple uses of a function with different arguments.
Different solutions to this conflict lead to
different representations of the storage, some of which are
examined by Accattoli and Barras~\cite{AccattoliB17} from the
perspective of time-cost analysis.
We recall a few solutions below that seem relevant to our token-guided
graph-rewriting.

One solution is to allow at most one assignment to each variable. This
is typically achieved by renaming bound variables during execution,
possibly symbolically. Examples for call-by-need evaluation are
Sestoft's abstract machines~\cite{Sestoft97}, and the storeless and
store-based abstract machines studied by Danvy and
Zerny~\cite{DanvyZ13}.
Our graph-rewriting abstract machine gives another example, as shown
by the simulation of the sub-machine semantics that resembles the
storeless abstract machine mentioned above. Variable renaming is
trivial in our machine, thanks to the use of graphs in which variables
are represented by mere edges.

Another solution is to allow multiple assignments to a variable,
with restricted visibility. The common approach is to pair a sub-term
with its own ``environment'' that maps its free variables to their
assigned computations, forming a so-called ``closure.'' Conflicting
assignments are distributed to distinct localised
environments. Examples include Cregut's lazy variant~\cite{Cregut07}
of Krivine's abstract machine for call-by-need evaluation, and
Landin's SECD machine~\cite{Landin64} for call-by-value evaluation.
Fern{\'{a}}ndez and Siafakas~\cite{FernandezS09} refine this approach
for call-by-name and call-by-value evaluations,
based on closed reduction~\cite{FernandezMS05}, which restricts
beta-reduction to closed function arguments.
This suggests that the approach with localised environments can be
modelled in our setting by implementing closed reduction. The
implementation would require an extension of rewrite transitions
and a different strategy to trigger them, namely to eliminate
auxiliary doors of a $\oc$-box.

Additionally, Fern{\'{a}}ndez and Siafakas~\cite{FernandezS09} propose
another approach to multiple assignments, in which multiple
assignments are augmented with binary strings so that each occurrence
of a variable can only refer to one of them. This approach is based on
a GoI-style token-passing abstract machine for call-by-value
evaluation, designed by Fern{\'{a}}ndez and Mackie~\cite{FernandezM02}.
The machine keeps the underlying graph fixed during execution and
allows the token to jump along a path in the graph. It therefore
recovers time efficiency, although no quantitative analysis is
provided.
Jumps can be seen as a form of graph rewriting that eliminates nodes.
Some jumps are to or from edges with an index that are effectively
``virtual'' copies of edges.

To wrap up, we presented a graph-rewriting abstract machine, with token passing as
a guide, that can time-efficiently implement three evaluation
strategies
that have different control over caching intermediate results.
The idea of using the token as a guide of graph rewriting was also proposed
by Sinot~\cite{Sinot05,Sinot06} for interaction nets.
He shows how using a token can make the rewriting system implement the
call-by-name, call-by-need and call-by-value evaluation strategies.
Our development in this work can be seen as a realisation of the
rewriting system as an abstract machine, in particular with explicit
control over copying sub-graphs.
The GoI-style token passing itself has been adapted to implement the
call-by-value evaluation strategy.
Apart from the abstract machine with jumps~\cite{FernandezM02} already
mentioned, known
adaptations~\cite{Schoepp14b,HoshinoMH14} commonly use the CPS
transformation~\cite{Plotkin75}, with the focus on
correctness. However this method naively leads to an abstract machine
with inefficient overhead cost~\cite{HoshinoMH14}. 

The token-guided graph rewriting is a flexible framework in which
we can carry out the study of space-time trade-off in abstract machines
for various evaluation strategies of the lambda-calculus. Starting
with the previous work~\cite{MuroyaG17} and continuing with the
present work, our focus
was primarily on time-efficiency. This is to 
complement existing work on GoI-style operational semantics which
usually achieves space-efficiency, and also to confirm that
introduction of graph rewriting to the semantics does not bring in any
hidden inefficiencies.
We believe that more refined strategies of interleaving token 
routing and graph reduction can be formulated to serve particular
objectives in the space-time execution efficiency trade-off.

\paragraph{\textbf{Acknowledgement}}
We thank Steven Cheung for helping us implement the on-line visualiser.



\bibliographystyle{eptcs}
\bibliography{ref}

\begin{thebibliography}{10}
\providecommand{\bibitemdeclare}[2]{}
\providecommand{\surnamestart}{}
\providecommand{\surnameend}{}
\providecommand{\urlprefix}{Available at }
\providecommand{\url}[1]{\texttt{#1}}
\providecommand{\href}[2]{\texttt{#2}}
\providecommand{\urlalt}[2]{\href{#1}{#2}}
\providecommand{\doi}[1]{doi:\urlalt{http://dx.doi.org/#1}{#1}}
\providecommand{\bibinfo}[2]{#2}

\bibitemdeclare{inproceedings}{Accattoli16}
\bibitem{Accattoli16}
\bibinfo{author}{Beniamino \surnamestart Accattoli\surnameend}
  (\bibinfo{year}{2017}): \emph{\bibinfo{title}{The complexity of abstract
  machines}}.
\newblock In: {\sl \bibinfo{booktitle}{WPTE 2016}}, {\sl
  \bibinfo{series}{EPTCS}} \bibinfo{volume}{235}, pp. \bibinfo{pages}{1--15},
  \doi{10.4204/EPTCS.235.1}.

\bibitemdeclare{inproceedings}{AccattoliBM14}
\bibitem{AccattoliBM14}
\bibinfo{author}{Beniamino \surnamestart Accattoli\surnameend},
  \bibinfo{author}{Pablo \surnamestart Barenbaum\surnameend} \&
  \bibinfo{author}{Damiano \surnamestart Mazza\surnameend}
  (\bibinfo{year}{2014}): \emph{\bibinfo{title}{Distilling abstract machines}}.
\newblock In: {\sl \bibinfo{booktitle}{ICFP 2014}}, \bibinfo{publisher}{{ACM}},
  pp. \bibinfo{pages}{363--376}, \doi{10.1145/2628136.2628154}.

\bibitemdeclare{inproceedings}{AccattoliB17}
\bibitem{AccattoliB17}
\bibinfo{author}{Beniamino \surnamestart Accattoli\surnameend} \&
  \bibinfo{author}{Bruno \surnamestart Barras\surnameend}
  (\bibinfo{year}{2017}): \emph{\bibinfo{title}{Environments and the complexity
  of abstract machines}}.
\newblock In: {\sl \bibinfo{booktitle}{PPDP 2017}}, \bibinfo{publisher}{{ACM}},
  pp. \bibinfo{pages}{4--16}, \doi{10.1145/3131851.3131855}.

\bibitemdeclare{article}{AccattoliDL16}
\bibitem{AccattoliDL16}
\bibinfo{author}{Beniamino \surnamestart Accattoli\surnameend} \&
  \bibinfo{author}{Ugo \surnamestart {Dal Lago}\surnameend}
  (\bibinfo{year}{2016}): \emph{\bibinfo{title}{(Leftmost-outermost) beta
  reduction is invariant, indeed}}.
\newblock {\sl \bibinfo{journal}{Logical Methods in Comp. Sci.}}
  \bibinfo{volume}{12}(\bibinfo{number}{1}), \doi{10.2168/LMCS-12(1:4)2016}.

\bibitemdeclare{inproceedings}{AccattoliG09}
\bibitem{AccattoliG09}
\bibinfo{author}{Beniamino \surnamestart Accattoli\surnameend} \&
  \bibinfo{author}{Stefano \surnamestart Guerrini\surnameend}
  (\bibinfo{year}{2009}): \emph{\bibinfo{title}{Jumping boxes}}.
\newblock In: {\sl \bibinfo{booktitle}{CSL 2009}}, {\sl \bibinfo{series}{Lect.
  Notes Comp. Sci.}} \bibinfo{volume}{5771}, \bibinfo{publisher}{Springer}, pp.
  \bibinfo{pages}{55--70}, \doi{10.1007/978-3-642-04027-6_7}.

\bibitemdeclare{inproceedings}{AccattoliSC14}
\bibitem{AccattoliSC14}
\bibinfo{author}{Beniamino \surnamestart Accattoli\surnameend} \&
  \bibinfo{author}{Claudio \surnamestart {Sacerdoti Coen}\surnameend}
  (\bibinfo{year}{2014}): \emph{\bibinfo{title}{On the value of variables}}.
\newblock In: {\sl \bibinfo{booktitle}{WoLLIC 2014}}, {\sl
  \bibinfo{series}{Lect. Notes Comp. Sci.}} \bibinfo{volume}{8652},
  \bibinfo{publisher}{Springer}, pp. \bibinfo{pages}{36--50},
  \doi{10.1007/978-3-662-44145-9_3}.

\bibitemdeclare{article}{Cregut07}
\bibitem{Cregut07}
\bibinfo{author}{Pierre \surnamestart Cr{\'{e}}gut\surnameend}
  (\bibinfo{year}{2007}): \emph{\bibinfo{title}{Strongly reducing variants of
  the {Krivine} abstract machine}}.
\newblock {\sl \bibinfo{journal}{Higher-Order and Symbolic Computation}}
  \bibinfo{volume}{20}(\bibinfo{number}{3}), pp. \bibinfo{pages}{209--230},
  \doi{10.1007/s10990-007-9015-z}.

\bibitemdeclare{article}{DanosR96}
\bibitem{DanosR96}
\bibinfo{author}{Vincent \surnamestart Danos\surnameend} \&
  \bibinfo{author}{Laurent \surnamestart Regnier\surnameend}
  (\bibinfo{year}{1996}): \emph{\bibinfo{title}{Reversible, irreversible and
  optimal lambda-machines}}.
\newblock {\sl \bibinfo{journal}{Elect. Notes in Theor. Comp. Sci.}}
  \bibinfo{volume}{3}, pp. \bibinfo{pages}{40--60},
  \doi{10.1016/S1571-0661(05)80402-5}.

\bibitemdeclare{article}{DanvyMMZ12}
\bibitem{DanvyMMZ12}
\bibinfo{author}{Olivier \surnamestart Danvy\surnameend},
  \bibinfo{author}{Kevin \surnamestart Millikin\surnameend},
  \bibinfo{author}{Johan \surnamestart Munk\surnameend} \& \bibinfo{author}{Ian
  \surnamestart Zerny\surnameend} (\bibinfo{year}{2012}):
  \emph{\bibinfo{title}{On inter-deriving small-step and big-step semantics: a
  case study for storeless call-by-need evaluation}}.
\newblock {\sl \bibinfo{journal}{Theor. Comp. Sci.}} \bibinfo{volume}{435}, pp.
  \bibinfo{pages}{21--42}, \doi{10.1016/j.tcs.2012.02.023}.

\bibitemdeclare{inproceedings}{DanvyZ13}
\bibitem{DanvyZ13}
\bibinfo{author}{Olivier \surnamestart Danvy\surnameend} \&
  \bibinfo{author}{Ian \surnamestart Zerny\surnameend} (\bibinfo{year}{2013}):
  \emph{\bibinfo{title}{A synthetic operational account of call-by-need
  evaluation}}.
\newblock In: {\sl \bibinfo{booktitle}{PPDP 2013}}, \bibinfo{publisher}{{ACM}},
  pp. \bibinfo{pages}{97--108}, \doi{10.1145/2505879.2505898}.

\bibitemdeclare{inproceedings}{FernandezM02}
\bibitem{FernandezM02}
\bibinfo{author}{Maribel \surnamestart Fern{\'{a}}ndez\surnameend} \&
  \bibinfo{author}{Ian \surnamestart Mackie\surnameend} (\bibinfo{year}{2002}):
  \emph{\bibinfo{title}{Call-by-value lambda-graph rewriting without
  rewriting}}.
\newblock In: {\sl \bibinfo{booktitle}{ICGT 2002}}, {\sl
  \bibinfo{series}{LNCS}} \bibinfo{volume}{2505},
  \bibinfo{publisher}{Springer}, pp. \bibinfo{pages}{75--89},
  \doi{10.1007/3-540-45832-8_8}.

\bibitemdeclare{article}{FernandezMS05}
\bibitem{FernandezMS05}
\bibinfo{author}{Maribel \surnamestart Fern{\'{a}}ndez\surnameend},
  \bibinfo{author}{Ian \surnamestart Mackie\surnameend} \&
  \bibinfo{author}{Fran{\c{c}}ois{-}R{\'{e}}gis \surnamestart Sinot\surnameend}
  (\bibinfo{year}{2005}): \emph{\bibinfo{title}{Closed reduction: explicit
  substitutions without alpha-conversion}}.
\newblock {\sl \bibinfo{journal}{Math. Struct. in Comp. Sci.}}
  \bibinfo{volume}{15}(\bibinfo{number}{2}), pp. \bibinfo{pages}{343--381},
  \doi{10.1017/S0960129504004633}.

\bibitemdeclare{article}{FernandezS09}
\bibitem{FernandezS09}
\bibinfo{author}{Maribel \surnamestart Fern{\'{a}}ndez\surnameend} \&
  \bibinfo{author}{Nikolaos \surnamestart Siafakas\surnameend}
  (\bibinfo{year}{2009}): \emph{\bibinfo{title}{New developments in environment
  machines}}.
\newblock {\sl \bibinfo{journal}{Elect. Notes in Theor. Comp. Sci.}}
  \bibinfo{volume}{237}, pp. \bibinfo{pages}{57--73},
  \doi{10.1016/j.entcs.2009.03.035}.

\bibitemdeclare{article}{Girard87LL}
\bibitem{Girard87LL}
\bibinfo{author}{Jean-Yves \surnamestart Girard\surnameend}
  (\bibinfo{year}{1987}): \emph{\bibinfo{title}{Linear logic}}.
\newblock {\sl \bibinfo{journal}{Theor. Comp. Sci.}} \bibinfo{volume}{50}, pp.
  \bibinfo{pages}{1--102}, \doi{10.1016/0304-3975(87)90045-4}.

\bibitemdeclare{inproceedings}{Girard89GoI1}
\bibitem{Girard89GoI1}
\bibinfo{author}{Jean-Yves \surnamestart Girard\surnameend}
  (\bibinfo{year}{1989}): \emph{\bibinfo{title}{Geometry of {Interaction} {I}:
  interpretation of system {F}}}.
\newblock In: {\sl \bibinfo{booktitle}{Logic Colloquium 1988}}, {\sl
  \bibinfo{series}{Studies in Logic \& Found. Math.}} \bibinfo{volume}{127},
  \bibinfo{publisher}{Elsevier}, pp. \bibinfo{pages}{221--260},
  \doi{10.1016/S0049-237X(08)70271-4}.

\bibitemdeclare{inproceedings}{HoshinoMH14}
\bibitem{HoshinoMH14}
\bibinfo{author}{Naohiko \surnamestart Hoshino\surnameend},
  \bibinfo{author}{Koko \surnamestart Muroya\surnameend} \&
  \bibinfo{author}{Ichiro \surnamestart Hasuo\surnameend}
  (\bibinfo{year}{2014}): \emph{\bibinfo{title}{Memoryful {Geometry} of
  {Interaction}: from coalgebraic components to algebraic effects}}.
\newblock In: {\sl \bibinfo{booktitle}{CSL-LICS 2014}},
  \bibinfo{publisher}{ACM}, pp. \bibinfo{pages}{52:1--52:10},
  \doi{10.1145/2603088.2603124}.

\bibitemdeclare{article}{KissingerPhD}
\bibitem{KissingerPhD}
\bibinfo{author}{Aleks \surnamestart Kissinger\surnameend}
  (\bibinfo{year}{2012}): \emph{\bibinfo{title}{Pictures of processes:
  automated graph rewriting for monoidal categories and applications to quantum
  computing}}.
\newblock {\sl \bibinfo{journal}{CoRR}} \bibinfo{volume}{abs/1203.0202}.
\newblock \urlprefix\url{http://arxiv.org/abs/1203.0202}.

\bibitemdeclare{article}{Landin64}
\bibitem{Landin64}
\bibinfo{author}{Peter \surnamestart Landin\surnameend} (\bibinfo{year}{1964}):
  \emph{\bibinfo{title}{The mechanical evaluation of expressions}}.
\newblock {\sl \bibinfo{journal}{The Comp. Journ.}}
  \bibinfo{volume}{6}(\bibinfo{number}{4}), pp. \bibinfo{pages}{308--320},
  \doi{10.1093/comjnl/6.4.308}.

\bibitemdeclare{inproceedings}{Mackie95}
\bibitem{Mackie95}
\bibinfo{author}{Ian \surnamestart Mackie\surnameend} (\bibinfo{year}{1995}):
  \emph{\bibinfo{title}{The {Geometry} of {Interaction} machine}}.
\newblock In: {\sl \bibinfo{booktitle}{POPL 1995}}, \bibinfo{publisher}{ACM},
  pp. \bibinfo{pages}{198--208}, \doi{10.1145/199448.199483}.

\bibitemdeclare{article}{MaraistOTW99}
\bibitem{MaraistOTW99}
\bibinfo{author}{John \surnamestart Maraist\surnameend},
  \bibinfo{author}{Martin \surnamestart Odersky\surnameend},
  \bibinfo{author}{David~N. \surnamestart Turner\surnameend} \&
  \bibinfo{author}{Philip \surnamestart Wadler\surnameend}
  (\bibinfo{year}{1999}): \emph{\bibinfo{title}{Call-by-name, call-by-value,
  call-by-need and the linear lambda calculus}}.
\newblock {\sl \bibinfo{journal}{Theor. Comp. Sci.}}
  \bibinfo{volume}{228}(\bibinfo{number}{1-2}), pp. \bibinfo{pages}{175--210},
  \doi{10.1016/S0304-3975(98)00358-2}.

\bibitemdeclare{inproceedings}{Mellies06}
\bibitem{Mellies06}
\bibinfo{author}{Paul{-}Andr{\'{e}} \surnamestart Melli{\`{e}}s\surnameend}
  (\bibinfo{year}{2006}): \emph{\bibinfo{title}{Functorial boxes in string
  diagrams}}.
\newblock In: {\sl \bibinfo{booktitle}{CSL 2006}}, {\sl \bibinfo{series}{Lect.
  Notes Comp. Sci.}} \bibinfo{volume}{4207}, \bibinfo{publisher}{Springer}, pp.
  \bibinfo{pages}{1--30}, \doi{10.1007/11874683_1}.

\bibitemdeclare{inproceedings}{MuroyaG17}
\bibitem{MuroyaG17}
\bibinfo{author}{Koko \surnamestart Muroya\surnameend} \&
  \bibinfo{author}{Dan~R. \surnamestart Ghica\surnameend}
  (\bibinfo{year}{2017}): \emph{\bibinfo{title}{The dynamic {Geometry} of
  {Interaction} machine: a call-by-need graph rewriter}}.
\newblock In: {\sl \bibinfo{booktitle}{CSL 2017}}, {\sl
  \bibinfo{series}{LIPIcs}}~\bibinfo{volume}{82}, pp.
  \bibinfo{pages}{32:1--32:15}, \doi{10.4230/LIPIcs.CSL.2017.32}.

\bibitemdeclare{article}{Plotkin75}
\bibitem{Plotkin75}
\bibinfo{author}{Gordon \surnamestart Plotkin\surnameend}
  (\bibinfo{year}{1975}): \emph{\bibinfo{title}{Call-by-name, call-by-value and
  the lambda-calculus}}.
\newblock {\sl \bibinfo{journal}{Theor. Comp. Sci.}}
  \bibinfo{volume}{1}(\bibinfo{number}{2}), pp. \bibinfo{pages}{125--259},
  \doi{10.1016/0304-3975(75)90017-1}.

\bibitemdeclare{inproceedings}{Schoepp14b}
\bibitem{Schoepp14b}
\bibinfo{author}{Ulrich \surnamestart Sch{\"{o}}pp\surnameend}
  (\bibinfo{year}{2014}): \emph{\bibinfo{title}{Call-by-value in a basic logic
  for interaction}}.
\newblock In: {\sl \bibinfo{booktitle}{APLAS 2014}}, {\sl
  \bibinfo{series}{Lect. Notes Comp. Sci.}} \bibinfo{volume}{8858},
  \bibinfo{publisher}{Springer}, pp. \bibinfo{pages}{428--448},
  \doi{10.1007/978-3-319-12736-1_23}.

\bibitemdeclare{article}{Sestoft97}
\bibitem{Sestoft97}
\bibinfo{author}{Peter \surnamestart Sestoft\surnameend}
  (\bibinfo{year}{1997}): \emph{\bibinfo{title}{Deriving a lazy abstract
  machine}}.
\newblock {\sl \bibinfo{journal}{J. Funct. Program.}}
  \bibinfo{volume}{7}(\bibinfo{number}{3}), pp. \bibinfo{pages}{231--264},
  \doi{10.1017/S0956796897002712}.

\bibitemdeclare{inproceedings}{Sinot05}
\bibitem{Sinot05}
\bibinfo{author}{Fran{\c{c}}ois{-}R{\'{e}}gis \surnamestart Sinot\surnameend}
  (\bibinfo{year}{2005}): \emph{\bibinfo{title}{Call-by-name and call-by-value
  as token-passing interaction nets}}.
\newblock In: {\sl \bibinfo{booktitle}{TLCA 2005}}, {\sl \bibinfo{series}{Lect.
  Notes Comp. Sci.}} \bibinfo{volume}{3461}, \bibinfo{publisher}{Springer}, pp.
  \bibinfo{pages}{386--400}, \doi{10.1007/11417170_28}.

\bibitemdeclare{article}{Sinot06}
\bibitem{Sinot06}
\bibinfo{author}{Fran{\c{c}}ois{-}R{\'{e}}gis \surnamestart Sinot\surnameend}
  (\bibinfo{year}{2006}): \emph{\bibinfo{title}{Call-by-need in token-passing
  nets}}.
\newblock {\sl \bibinfo{journal}{Math. Struct. in Comp. Sci.}}
  \bibinfo{volume}{16}(\bibinfo{number}{4}), pp. \bibinfo{pages}{639--666},
  \doi{10.1017/S0960129506005408}.

\end{thebibliography}
\end{document}